\documentclass[12pt,table]{article}
\usepackage{graphicx} 
\usepackage{epstopdf}
\usepackage{amsthm,amsmath,amssymb, amsfonts, amscd, xspace, pifont}
\usepackage{epsfig, psfrag,pstricks}
\usepackage{times}
\usepackage[authoryear,round]{natbib}
\usepackage{algorithm}
\usepackage{algpseudocode}
\usepackage{booktabs}
\usepackage{multirow}
\usepackage{makecell}
\usepackage{array}
\usepackage{tabularx}
\usepackage{setspace}
\usepackage{subcaption} 

\newtheorem{assump}{Assumption}

\newtheorem{proposition}{Proposition}
\providecommand{\keywords}[1]{\textbf{\textit{Keywords: }} #1}

\def\bepsilon{\mbox{\boldmath $\epsilon$}}  

\def\bmu{\mbox{\boldmath $\mu$}}

\def\bR{\mathbb{R}}
\def\bE{\mathbb{E}}   

\def\bN{\mathbb{N}}

\def\bb{{\bf b}}

\def\bY{{\bf Y}}

\def\bw{{\bf w}}

\def\bv{{\bf v}}

\def\bI{{\bf I}}
\def\bx{{\bf x}}
\def\by{{\bf y}}
\def\bz{{\bf z}}

\def\nn{\nonumber}
\def\v2{\vspace{0.2in}}

\def\diff{\mathrm{d}}
\def\var{\text{var}}

\newcommand{\ee}{{\rm e}}

\numberwithin{equation}{section}
\begin{document}
\setcounter{page}{1}
\baselineskip=15pt
\footskip=.3in
\parskip=5pt
\parindent=15pt

\title{Annealed Langevin Monte Carlo for Flow ODE Sampling}

\date{}

\author{Hanwen Huang \\\\
    {\it Department of Biostatistics, Data Science and Epidemiology}\\
    {\it Medical College of Georgia}\\
    {\it Augusta University, Augusta, GA 30912}\\
     hhuang1@augusta.edu}

\maketitle

\begin{abstract}
We propose Annealed Langevin Monte Carlo for Flow ODE Sampling (ALMC-ODE), a method for generating samples from unnormalized target distributions, with a particular emphasis on multimodal densities that are challenging for standard Markov chain Monte Carlo methods. ALMC-ODE is based on a probability-flow ordinary differential equation (ODE) derived from stochastic interpolants, which continuously transports a standard Gaussian reference distribution at $t = 0$ to the target distribution $\rho$ at $t = 1$. The key innovation lies in an annealed Langevin Markov chain that evolves through a sequence of intermediate distributions bridging the reference and the target. The resulting importance-weighted particles, reweighted via a Jarzynski-based scheme, yield a low-variance estimator of the velocity field governing the ODE. On the theoretical side, we establish a Jarzynski-type reweighting identity for general time-inhomogeneous transition kernels, characterize the optimal backward kernel that minimizes the variance of the importance weights, and prove an $\mathcal{O}(1/n)$ mean squared error bound for the resulting velocity-field estimator. Numerical experiments on challenging benchmarks, including Gaussian mixture models and a 64-dimensional Allen--Cahn field system, demonstrate that ALMC-ODE significantly outperforms both direct Monte Carlo ODE approaches and Hamiltonian Monte Carlo when applied to highly multimodal target distributions.
\end{abstract}


\section{Introduction}
We consider the problem of sampling from a target distribution $\rho\propto\ee^{-V}$ on $\bR^d$. This problem is central to computational statistics computational statistics \citep{liu2004monte,brooks2011handbook}, Bayesian inference \citep{gelman2013bayesian}, statistical physics \citep{newman1999monte}, finance \citep{dagpunar2007simulation}, and related areas, and has been studied extensively in the literature \citep{chewi2022log}. A standard approach is Markov Chain Monte Carlo (MCMC), encompassing algorithms such as Metropolis–Hastings \citep{metropolis1953equation,HASTINGS}, Gibbs sampling \citep{Gelfand01061990}, Langevin-based methods \citep{Exponential, Dalalyan}, and Hamiltonian Monte Carlo (HMC) \citep{duane1987,neal2011}. 

While MCMC methods are well understood for log-concave or unimodal targets, their performance deteriorates sharply in multimodal settings. In particular, the presence of large energy barriers induces metastability, leading to exponentially slow mixing and poor exploration. It is by now well established that even optimally tuned HMC and random-walk Metropolis algorithms can exhibit exponential mixing times in the barrier height \citep{mangoubi2018hamiltonian,ma2019sampling,dunson2020hastings,dong2022spectral}. This motivates the development of sampling schemes capable of efficiently traversing complex energy landscapes. Existing approaches include tempering and annealing methods \citep{Marinari1992,neal2001annealed}, equi-energy sampling \citep{kou2006}, adaptive biasing techniques \citep{WangLandau2001,doi:10.1073/pnas.202427399}, and flow-based or dynamical transport methods \citep{doi:10.1073/pnas.202427399,Hackett2021FlowbasedSF,doi:10.1137/23M1577584}.

In this work, we propose \emph{Annealed Langevin Monte Carlo for Flow ODE Sampling} (ALMC-ODE), a method that constructs a transport from a tractable reference measure to $\rho$ via the probability-flow ODE associated with a stochastic interpolant. The key computational bottleneck in such approaches is the estimation of the velocity field. Direct Monte Carlo estimators \citep{ding2023samplingfollmerflow} are known to suffer from high variance, particularly in high dimensions. To address this, we introduce an annealed scheme \citep{gelfand1990sampling,neal2001annealed,song2019generative,song2020improved,zilberstein2022annealed,zilberstein2024solving} that defines a sequence of intermediate distributions $\{\pi_k\}_{k=0}^K$ connecting a reference $\pi_0$ (e.g., Gaussian) to $\pi_K=\rho$. We simulate an annealed Langevin Markov chain along this path and construct an importance-weighted estimator of the velocity field. The resulting estimator leverages both local equilibration and global reweighting, yielding substantially reduced variance relative to naive Monte Carlo integration.

Our contributions are as follows:
\begin{itemize}
\item We extend the Jarzynski-type reweighting framework of \cite{Carbone, cuin2025learning} to general time-inhomogeneous Markov chains and show that the variance-minimizing importance weights coincide with those of the standard importance sampling estimator.
\item We introduce a consistent empirical estimator of the optimal weights along the annealing path.
\item We establish a non-asymptotic mean-squared error bound of order $\mathcal{O}(1/n)$ for the resulting velocity field estimator.
\item We demonstrate empirically, on high-dimensional and multimodal benchmarks (including Gaussian mixtures and Allen–Cahn-type field models), that ALMC-ODE outperforms both direct Monte Carlo ODE methods and HMC in terms of sampling efficiency.
\end{itemize}

\paragraph{Related work.}
Flow-based constructions of transport maps via probability-flow ODEs and diffusion models have recently received significant attention \citep{albergo2023stochastic,albergo2023building,lipman2023flow}. A central challenge in these approaches is the accurate estimation of the velocity field (or equivalently, the score function). Existing methods based on importance sampling \citep{huang2024faster,51} or rejection sampling \citep{he2024zeroth} degrade in high dimensions due to variance explosion. Langevin-based estimators improve scalability by exploiting local geometric information, but typically require repeated sampling from intermediate conditional distributions  \citep{39,huang2024reverse,huang2024faster}. In contrast, our method integrates annealing and importance reweighting within a single sequential Monte Carlo–type procedure, thereby amortizing sampling effort across the path. Closely related is the work of \cite{Duan2026SamplingVS}, which employs Langevin dynamics to estimate the velocity field at each time point; however, their approach does not exploit annealing and thus incurs a substantially higher computational cost due to repeated independent simulations.

The remainder of the paper is organized as follows. Section~\ref{method} introduces the stochastic interpolant and probability-flow ODE framework and reviews the Jarzynski identity underlying our reweighting scheme. Section~\ref{sec:algorithm} presents the ALMC-ODE algorithm. Section~\ref{mse} provides the error analysis. Section~\ref{numerical} contains numerical experiments, and Section~\ref{conclusion} concludes.
\section{Sampling via Flow ODE}\label{method}

This section develops the theoretical foundation of our approach. We begin in
Section~\ref{condsb} by deriving a probability-flow ODE that continuously transports
samples from a standard Gaussian to the target $\rho$. The central computational
challenge is estimating the velocity field that governs this ODE. In
Section~\ref{almc} we show that this reduces to a weighted expectation over samples from
an intermediate distribution, and in Section~\ref{sec:sbp} we describe how annealed
Langevin dynamics combined with a Jarzynski-based reweighting scheme can efficiently
supply those samples.

\subsection{Stochastic Interpolant and Probability Flow ODE}\label{condsb}

Our goal is to construct a deterministic map that transports samples from a tractable
Gaussian distribution to the target $\rho(\bx)$. The key idea is to design a
time-dependent stochastic process that interpolates between a Gaussian source at $t = 0$
and the target at $t = 1$. Following \cite{albergo2023stochastic,albergo2023building}, we
define the \emph{stochastic interpolant}
\begin{equation}
    \label{eq:vec-interp}
    \bx_t = \alpha(t)\,\bz + \beta(t)\,\bx_1,
\end{equation}
where $\bz \sim N(0, \bI_d)$ and $\bx_1 \sim \rho$ are independent, and $\alpha(t)$,
$\beta(t)$ are smooth scalar functions satisfying the boundary conditions
\begin{eqnarray}\label{eq:abt}
\alpha(0) = 1, \quad \beta(0) = 0, \quad \alpha(1) = 0, \quad \beta(1) = 1,
\end{eqnarray}
so that $\bx_0 = \bz \sim N(0,\bI_d)$ and $\bx_1 \sim \rho$ as required.  Several
concrete choices of $(\alpha, \beta)$ satisfying these conditions are listed in
Table~\ref{tab:vec-interp} \citep{JMLR:v25:23-1515}.

\begin{table}[thb]
    \centering
    \begin{tabular}{|c|c|c|c|}\hline
    Type  & Linear & F{\"o}llmer & Trigonometric \\\hline
    $\alpha(t)$  & $1-t$ & $\sqrt{1 - t^2}$ & $\cos(\tfrac{\pi}{2}t)$ \\\hline
    $\beta(t)$  & $t$ & $t$   & $\sin(\tfrac{\pi}{2}t)$ \\\hline
    \end{tabular}
    \caption{Summary of various measure interpolants: the linear interpolant
    \citep{liu2023flow}, the F{\"o}llmer interpolant \citep{dai2023lipschitz}, and the
    trigonometric interpolant \citep{albergo2023building}.}
    \label{tab:vec-interp}
\vspace{-0.1cm}
\end{table}

It was shown in \cite{albergo2023stochastic,albergo2023building} that the marginal density
of $\bx_t$ defined in \eqref{eq:vec-interp} satisfies a continuity equation, and its
evolution is described by the following probability-flow ODE (initial value problem):
\begin{equation} \label{eq:ode-ivp0}
    \frac{\diff \bx_t}{\diff t} = \bv(t, \bx_t), \quad \bx_0 \sim N(0,\bI_d),
    \quad t \in [0, 1],
\end{equation}
where the velocity field is
\begin{eqnarray}\label{sbdm}
\bv(t,\bx) = \dot\alpha(t)\,\bE(\bz\,|\,\bx_t=\bx)
           + \dot\beta(t)\,\bE(\bx_1\,|\,\bx_t=\bx).
\end{eqnarray}
Here and throughout, $\dot\alpha$ and $\dot\beta$ denote derivatives with respect to $t$.
Crucially, at any time $t \in [0,1]$, the probability-flow ODE \eqref{eq:ode-ivp0} has
the same marginal density as the stochastic interpolant $\bx_t$ defined in
\eqref{eq:vec-interp}. In particular, simulating \eqref{eq:ode-ivp0} from $\bx_0 \sim
N(0,\bI_d)$ produces samples with distribution $\rho$ at $t = 1$, provided the velocity
field $\bv(t,\bx)$ can be evaluated.

\paragraph{Reducing the velocity field to a single conditional expectation.}
From \eqref{eq:vec-interp} we have $\bx = \alpha(t)\,\bE[\bz \mid \bx_t = \bx] +
\beta(t)\,\bE[\bx_1 \mid \bx_t = \bx]$, which allows us to eliminate the $\bz$-term from
\eqref{sbdm}:
\begin{equation}
    \label{eq:vf-expect-single-rmk}
    \bv(t, \bx) = \frac{\dot{\alpha}(t)}{\alpha(t)}\,\bx
      + \left( \dot{\beta}(t) - \frac{\dot{\alpha}(t)}{\alpha(t)}\beta(t) \right)
        \bE[\bx_1 \mid \bx_t = \bx], \quad t \in [0, 1).
\end{equation}
Thus, computing $\bv(t,\bx)$ reduces to evaluating the single conditional expectation
$\bE[\bx_1 \mid \bx_t = \bx]$.  Applying Bayes' theorem and using the fact that
$\bx_t \mid \bx_1 \sim N(\beta(t)\bx_1,\, \alpha(t)^2 \bI_d)$ (from
\eqref{eq:vec-interp}), this conditional expectation takes the integral form
\begin{eqnarray}\label{velocity_int}
\bE[\bx_1 \mid \bx_t=\bx]
  = \frac{\displaystyle\int\bx_1\exp\!\left\{-\frac{\|\bx-\beta(t)\bx_1\|^2}{2\alpha(t)^2}\right\}\rho(\bx_1)\,\diff\bx_1}
         {\displaystyle\int\exp\!\left\{-\frac{\|\bx-\beta(t)\bx_1\|^2}{2\alpha(t)^2}\right\}\rho(\bx_1)\,\diff\bx_1}.
\end{eqnarray}
The integral in \eqref{velocity_int} is generally intractable and must be approximated
numerically. A notable exception is when $\rho$ is a Gaussian mixture, in which case it
admits a closed form. In the general case, \cite{ding2023samplingfollmerflow} proposed
Monte Carlo approximation using samples drawn directly from $N(0, \bI_d)$ as the
proposal. However, when $\rho$ differs substantially from a Gaussian, this approach
suffers from high estimation variance due to poor overlap between the Gaussian proposal
and the target \citep{vargas2022bayesian}. In our experiments, this instability was
severe: the method failed to produce reliable results even for low-dimensional multimodal
targets (see Section~\ref{numerical}). This motivates our approach, described next, which
replaces the Gaussian proposal with a sequence of annealed Langevin samplers that
progressively adapt to $\rho$.

\subsection{Annealed Langevin Monte Carlo}\label{almc}

\paragraph{From generative learning to unnormalized sampling.}
In the generative-learning setting, suppose we have access to $n$ i.i.d.\ samples
$\{\bx_1^{(i)}\}_{i=1}^n$ drawn from $\rho$.  The integral in \eqref{velocity_int} can
then be approximated by replacing the expectation over $\rho$ with its empirical
counterpart, $\int f(\bx_1)\rho(\bx_1)\,\diff\bx_1 \approx \frac{1}{n}\sum_{i=1}^n
f(\bx_1^{(i)})$ for any measurable $f$.  This yields the direct estimator
\begin{equation}
    \label{eq:velocity-field}
    \hat{\bv}(t, \bx) = \frac{\dot{\alpha}(t)}{\alpha(t)}\,\bx
      + \left( \dot{\beta}(t) - \frac{\dot{\alpha}(t)}{\alpha(t)}\beta(t) \right)
        \frac{\displaystyle\sum_{i=1}^n\bx_1^{(i)}
              \exp\!\left\{-\frac{\|\bx-\beta(t)\bx_1^{(i)}\|^2}{2\alpha(t)^2}\right\}}
             {\displaystyle\sum_{i=1}^n
              \exp\!\left\{-\frac{\|\bx-\beta(t)\bx_1^{(i)}\|^2}{2\alpha(t)^2}\right\}}.
\end{equation}
In practice, however, the target $\rho$ is available only up to a normalizing
constant—direct sampling from it is the very problem we aim to solve.  We therefore turn
to importance sampling: draw $\{\bx_1^{(i)}\}_{i=1}^n$ from a tractable proposal
$\hat\rho$ and attach importance weights $w^{(i)} = \rho(\bx_1^{(i)})/\hat\rho(\bx_1^{(i)})$.
This yields the importance-weighted velocity estimator
\begin{equation}
    \label{eq:velocity-field-IS}
    \hat{\bv}(t, \bx) = \frac{\dot{\alpha}(t)}{\alpha(t)}\,\bx
      + \left( \dot{\beta}(t) - \frac{\dot{\alpha}(t)}{\alpha(t)}\beta(t) \right)
        \frac{\displaystyle\sum_{i=1}^n \bx_1^{(i)}
              \exp\!\left\{-\frac{\|\bx-\beta(t)\bx_1^{(i)}\|^2}{2\alpha(t)^2}\right\}
              w^{(i)}}
             {\displaystyle\sum_{i=1}^n
              \exp\!\left\{-\frac{\|\bx-\beta(t)\bx_1^{(i)}\|^2}{2\alpha(t)^2}\right\}
              w^{(i)}}.
\end{equation}
The quality of \eqref{eq:velocity-field-IS} depends critically on the weight variance:
when $\hat\rho$ and $\rho$ have limited overlap, the weights degenerate and the estimator
becomes unreliable.  This places a strong premium on choosing a proposal $\hat\rho$ that
closely approximates $\rho$ while remaining easy to sample from.

\paragraph{Langevin dynamics as a sampling engine.}
A natural candidate for the proposal $\hat\rho$ is the stationary distribution of a
Langevin diffusion.  Recall that the Langevin SDE with invariant measure
$\rho \propto \ee^{-V}$ is
\begin{eqnarray}\label{langevin}
\diff\bY_s = -\nabla V(\bY_s)\,\diff s + \sqrt{2}\,\diff\bw_s,
\quad s \in [0,\infty),
\quad \bY_0 \sim \pi_0,
\end{eqnarray}
where $(\bw_s)_{s\geq 0}$ is a standard Brownian motion in $\bR^d$, $\pi_0$ is an
arbitrary initialization distribution, and we write $\bY_s$ (rather than $\bx_t$) to
distinguish this Langevin chain from the stochastic interpolant $\bx_t$ in
\eqref{eq:vec-interp}.  Under suitable regularity conditions, the law of $\bY_s$ converges
to $\rho$ as $s \to \infty$.  However, for complex multimodal targets, the energy
landscape induced by $V$ may contain many local minima separated by high barriers, causing
the diffusion to become trapped and mix poorly.

\paragraph{Annealing for improved mixing.}
To overcome this mixing problem, we employ the \emph{annealed Langevin diffusion} (ALD)
framework, which gradually transforms a tractable initial distribution $\pi_0$ into the
target $\rho$ through a sequence of intermediate distributions.  Formally, we introduce a
path $(\pi_t)_{t\in[0,1]}$ interpolating between $\pi_0$ and $\pi_1 = \rho$.  A common
choice is the geometric path $\pi_t(\bx) \propto \pi_0(\bx)^{1-\lambda(t)}\rho(\bx)^{\lambda(t)}$ with $\lambda(0)=0$ and $\lambda(1)=1$, though
other interpolations are possible.  Reparametrizing as $\tilde{\pi}_t = \pi_{t/T}$ for $t
\in [0,T]$ with $T > 0$ a sufficiently large time horizon, we consider the
time-inhomogeneous SDE
\begin{eqnarray}\label{annealed}
\diff\bY_t = \nabla\log\tilde{\pi}_t(\bY_t)\,\diff t + \sqrt{2}\,\diff\bw_t,
\quad t \in [0,T].
\end{eqnarray}
Intuitively, when the target $\tilde{\pi}_t$ evolves slowly enough, the law of $\bY_t$
closely tracks $\tilde{\pi}_t$, so that $\bY_T$ provides an approximate sample from $\rho$.
The annealing path smooths the energy landscape in the early stages, facilitating
exploration across modes.

\paragraph{From ALD to a practical algorithm: ULA with Jarzynski correction.}
Two difficulties arise when discretizing \eqref{annealed}.  First, iterating the
unadjusted Langevin algorithm (ULA) under a slowly varying potential does \emph{not}
produce exact samples from the instantaneous target $\tilde{\pi}_t$.  Second, this bias
cannot be removed by a Metropolis–Hastings correction, because it stems not only from time
discretization but also from the continuous evolution of $(\tilde{\pi}_t)$.  To address both issues, \cite{Carbone, cuin2025learning} introduce the \emph{Jarzynski-adjusted Langevin algorithm} (JALA), which runs ULA under the time-dependent potential sequence — accepting the induced bias — and corrects the resulting particles via exponential reweighting derived from the Jarzynski equality. We extend this framework to general time-inhomogeneous Markov chains and refine the correction by adopting variance-minimizing importance weights in place of the standard Jarzynski weights. The remainder of this section makes this precise.

\subsection{Jarzynski Adjustment}\label{sec:sbp}

As noted above, ULA iterates under a time-varying potential do not produce exact samples
from the instantaneous target $\tilde{\pi}_t$—the law of the chain generally lags behind
due to non-equilibrium bias.  We correct for this using the Jarzynski equality from
nonequilibrium statistical physics \citep{PhysRevLett.78.2690}.

We state the result for a general time-inhomogeneous Markov chain.  Let
$\mu_k(\bx,\by)$ denote the forward transition kernel from $\bx$ to $\by$ at step $k$,
and let $\nu_k(\bx_k,\bx_{k-1})$ be a corresponding backward kernel satisfying
$\int\nu_k(\bx_k,\bx_{k-1})\,\diff\bx_{k-1}=1$.  We associate each step $k$ with a
target density
\begin{eqnarray}\nn
\tilde{\pi}_k(\bx)=Z_k^{-1}\ee^{-V_k(\bx)},
\qquad Z_k=\int \ee^{-V_k(\bx)}\,\diff\bx.
\end{eqnarray}
The following proposition extends Proposition 1 of \cite{cuin2025learning} to general transition kernels, providing a reweighting formula that exactly corrects the discrepancy between the law of the chain at step $k$ and the target $\tilde{\pi}_k$.
\begin{proposition}\label{prop1}
Let $\{\mu_k\}$ be any sequence of time-dependent Markov transition kernels.  Define the
chain $(\bx_k,A_k)\in\bR^d\times\bR$ recursively by
\begin{eqnarray}\label{iteration}
\left\{\begin{array}{ll}
\bx_{k} \sim \mu_{k}(\bx_{k-1},\cdot), & \bx_0\sim\pi_{0},\\[4pt]
A_{k} = A_{k-1} + V_{k-1}(\bx_{k-1}) - V_{k}(\bx_{k})
        + \log\dfrac{\nu_{k}(\bx_{k},\bx_{k-1})}{\mu_{k}(\bx_{k-1},\bx_{k})},
& A_0=0.
\end{array}\right.
\end{eqnarray}
Then, for any $k\in\bN$ and any measurable function $f:\bR^d\rightarrow\bR$,
\begin{eqnarray}\label{adjust}
\bE_{k}[f(\bx_k)]=\frac{\bE[f(\bx_k)\ee^{A_k}]}{\bE[\ee^{A_k}]},
\qquad Z_{k}=Z_{0}\,\bE[\ee^{A_k}],
\end{eqnarray}
where $\bE_{k}$ denotes expectation with respect to $\tilde{\pi}_k$, and the expectations
on the right-hand side are taken over the joint law of $(\bx_k,A_k)$ induced by the
iteration \eqref{iteration}.
\end{proposition}
See Appendix~\ref{proof1} for a proof. Proposition~\ref{prop1} holds for \emph{any} choice of forward and backward kernels for which the Radon–Nikodym derivative $P_k / Q_k$ is well defined, where $P_k = \pi_0(\bx_0)\prod_{q=1}^k \mu_q(\bx_{q-1}, \bx_q)$ and $Q_k = \tilde{\pi}_k(\bx_k)\prod_{q=0}^{k-1} \nu_q(\bx_{q+1}, \bx_q)$. The results of \cite{Carbone, Carbone2024, cuin2025learning, Carbone2025} are recovered as special cases; in particular, \cite{Carbone2025} also accommodates a general forward kernel.  The interacting particle Langevin algorithm \citep{akyildiz2025interactingparticlelangevinalgorithm, marks2025learninglatentenergybasedmodels} shares a similar spirit and likewise falls within this framework as a special case. Moreover, the well-known annealed importance sampling (AIS) method also arises as a special case, obtained by choosing the backward kernels as
\begin{eqnarray}\label{eq:ais-backward}
\nu^{\mathrm{ais}}_k(\bx_k, \bx_{k-1})
= \frac{\tilde{\pi}_{k+1}(\bx_k)\mu_{k+1}(\bx_k, \bx_{k+1})}{\tilde{\pi}_{k+1}(\bx_{k+1})}.
\end{eqnarray}
Although the AIS backward kernels \eqref{eq:ais-backward} are analytically convenient, they are generally suboptimal with respect to the variance of the resulting estimator \citep{doucet2022scorebaseddiffusionmeetsannealed}. The following proposition, which parallels Proposition 1 of \cite{doucet2022scorebaseddiffusionmeetsannealed}, characterizes the optimal backward kernel — that is, the one minimizing the variance of the reweighting estimator \eqref{adjust}.
\begin{proposition}\label{prop2}
Among all choices of backward kernel $\nu_k$ in \eqref{iteration}, the one that minimizes
the variance of $\exp(A_k)$ is
\begin{eqnarray}\label{back}
\nu_k^{\mathrm{opt}}(\bx_k,\bx_{k-1})
  = \frac{p_{k-1}(\bx_{k-1})\,\mu_k(\bx_{k-1},\bx_k)}{p_k(\bx_k)},
\end{eqnarray}
where $p_k(\bx_k) = \int\pi_0(\bx_0)\prod_{q=1}^k\{\mu_q(\bx_{q-1},\bx_q)\,\diff\bx_{q-1}\}$
denotes the marginal density of $\bx_k$ induced by the forward iteration \eqref{iteration}.
\end{proposition}
In words, the optimal backward kernel \eqref{back} corresponds to the time-reversal of
the forward chain: $\nu_k^{\mathrm{opt}}$ is initialized from $\tilde{\pi}_k$ and runs
the reversed dynamics of $\{\mu_q\}$.  In the ideal case $p_k = \tilde{\pi}_k$ the
optimal and forward kernels coincide, since the backward decomposition of $Q$ reduces to
the forward process.  Substituting \eqref{back} into $A_k$ yields a particularly clean
expression:
\begin{eqnarray}\nn
\exp(A_k)
  &=& \ee^{-V_k(\bx_k)+V_{0}(\bx_0)}
      \prod_{q=1}^k \frac{\nu_q(\bx_q,\bx_{q-1})}{\mu_q(\bx_{q-1},\bx_q)} \\\nn
  &=& \ee^{-V_k(\bx_k)+V_{0}(\bx_0)}
      \prod_{q=1}^k
      \frac{p_{q-1}(\bx_{q-1})\,\mu_q(\bx_{q-1},\bx_q)}
           {p_q(\bx_q)\,\mu_q(\bx_{q-1},\bx_q)}
   = \frac{Z_k}{Z_0}\cdot\frac{\tilde{\pi}_k(\bx_k)}{p_k(\bx_k)},
\end{eqnarray}
so \eqref{adjust} simplifies to the standard importance-sampling estimator with marginal
weights $w_k = \tilde{\pi}_k(\bx_k)/p_k(\bx_k)$:
\begin{eqnarray}\nn
\bE_{k}[f(\bx_k)] = \frac{\bE\!\left[f(\bx_k)\,w_k\right]}{\bE[w_k]},
\qquad w_k = \frac{\tilde{\pi}_k(\bx_k)}{p_k(\bx_k)}.
\end{eqnarray}

\paragraph{Estimating the forward density $p_k$.}
To make the importance weight $w_k$ computable, we need to evaluate $p_k(\bx_k)$—the
marginal density of $\bx_k$ under the forward chain.  By the Markov structure,
\begin{eqnarray}\nn
p_k(\bx_k) = \int p_{k-1}(\bx_{k-1})\,\mu_k(\bx_{k-1},\bx_k)\,\diff\bx_{k-1}.
\end{eqnarray}
Given $n$ particles $\{\bx_{k-1}^{(i)}\}_{i=1}^n$ from step $k-1$, this density can be
approximated as
\begin{eqnarray}\label{density}
\hat{p}_k(\bx_k) \approx \frac{1}{n}\sum_{i=1}^n \mu_k(\bx^{(i)}_{k-1},\bx_k).
\end{eqnarray}
Proposition~\ref{prop1} applies to any forward kernel $\mu_k$—including MCMC kernels, ULA,
or deterministic maps.  For our method, we use the Euler–Maruyama discretization of the
annealed Langevin SDE \eqref{annealed}, which updates the Langevin chain as
\begin{eqnarray}\label{discrete}
\bx_k = \bx_{k-1} - \delta_k\,\nabla V_k(\bx_{k-1}) + \sqrt{2\delta_k}\,\bepsilon_k,
\qquad \bepsilon_k \sim N(0,\bI_d),
\end{eqnarray}
where $\delta_k > 0$ is the step size and
\begin{equation}
    \nabla V_k(\bx_{k-1}) = (1 - \lambda(t_k))\, \bx_{k-1} - \lambda(t_k) \, \nabla \log \rho(\bx_{k-1}),
    \label{eq:drift}
\end{equation}
At $t_k = 0$ this reduces to $\bx_{k-1}$ (the gradient of the standard Gaussian potential $-\log\pi_0$), and at $t_k = 1$ it equals $-\nabla\log\rho(\bx_{k-1})$ (the gradient of the target potential).  Here we revert to the notation $\bx_k$ for the discretized Langevin iterates; the interpolant process $\bx_t$ from Section~\ref{condsb} does not appear in the discretization and there is no ambiguity.  The update \eqref{discrete} induces the Gaussian transition density
\begin{eqnarray}\label{ula}
\mu_k(\bx_{k-1},\bx_k)
= \frac{1}{(4\pi\delta_k)^{d/2}}
  \exp\!\left(
  -\frac{1}{4\delta_k}
  \left|\bx_k - \bx_{k-1} + \delta_k\,\nabla V_k(\bx_{k-1})\right|^2
  \right).
\end{eqnarray}
Substituting \eqref{ula} into \eqref{density}, the forward density at $\bx_k$ is
estimated as
\begin{eqnarray}\label{density2}
\hat{p}_k(\bx_k)
\approx \frac{1}{n}\sum_{i=1}^n
\frac{1}{(4\pi\delta_k)^{d/2}}
\exp\!\left(
-\frac{1}{4\delta_k}
\left|\bx_k - \bx_{k-1}^{(i)} + \delta_k\,\nabla V_k(\bx_{k-1}^{(i)})\right|^2
\right),
\end{eqnarray}
and the corresponding importance weight for particle $\bx_k$ is
\begin{eqnarray}\label{eweight}
\hat{w}_k = \frac{\tilde{\pi}_k(\bx_k)}{\hat{p}_k(\bx_k)}.
\end{eqnarray}

\section{Annealed Langevin Diffusion ODE Flow Sampling}\label{sec:algorithm}

We now describe the complete sampling procedure, which combines the two components
developed in Section~\ref{method}: (i) an annealed Langevin Monte Carlo (ALMC) phase that
generates importance-weighted particles approximating the target $\rho$, and (ii) an ODE
phase that uses those particles to estimate the velocity field and simulate the probability
flow \eqref{eq:ode-ivp0}.  The two phases are summarized in Algorithms~\ref{alg1}
and~\ref{alg2}, respectively.

\paragraph{Phase 1: ALMC sampling and reweighting (Algorithm~\ref{alg1}).}
We run $n$ particles forward through $K$ steps of ULA under the annealed potential
sequence $\{V_k\}_{k=0}^{K-1}$.  At each step $k$, importance weights $\{w_k^{(i)}\}$
are computed via \eqref{eweight} to correct for the bias of ULA relative to
$\tilde{\pi}_k$.  Since this is a sequential Monte Carlo (SMC) scheme, the normalized
log-weights can degenerate over iterations—a phenomenon tracked by the effective sample
size
\[
  \mathrm{ESS}_k
    = \frac{\bigl(\sum_{i=1}^n w_k^{(i)}\bigr)^2}{\sum_{i=1}^n (w_k^{(i)})^2}
    \in [1, n].
\]
Initially $\mathrm{ESS}_0 = n$, and it decreases as the annealing proceeds.  When
$\mathrm{ESS}_k$ falls below a threshold $C$, a resampling step is triggered to replenish
particle diversity.

\begin{algorithm}[H]
	\caption{ALMC: Generating weighted particles from the target $\rho$}
    \label{alg1}
	\begin{algorithmic}[1]
\State \textbf{Input:} potential sequence $\{V_k\}_{k=0}^{K-1}$, step size $\{\delta_k\}_{k=1}^{K} > 0$,
       number of particles $n \in \bN$, number of steps $K \in \bN$, ESS threshold $C > 0$.
\State \textbf{Initialize:} sample $\{\bx_0^{(i)}\}_{i=1}^n \sim N(0,\bI_d)$; set $w_0^{(i)} = 1$ for all $i$.
\For{$k = 0, 1, \ldots, K-1$}
  \State Draw $\bepsilon_{k}^{(i)} \sim N(0,\bI_{d})$ independently for each $i$.
  \State Update: $\bx_{k+1}^{(i)} = \bx_{k}^{(i)} - \delta_k\,\nabla V_k(\bx_k^{(i)}) + \sqrt{2\delta_k}\,\bepsilon_{k}^{(i)}$.
  \State Compute importance weights $\{w_{k+1}^{(i)}\}$ via \eqref{eweight}.
  \State If $\mathrm{ESS}_{k+1} < C$, resample $\{\bx_{k+1}^{(i)}\}$ with probabilities proportional to $\{w_{k+1}^{(i)}\}$ and reset weights to $1/n$.
\EndFor
\State \textbf{Output:} particles $\{\bx_{K}^{(i)}\}_{i=1}^n$ with importance weights $\{w_{K}^{(i)}\}_{i=1}^n$.
\end{algorithmic}
\end{algorithm}

\paragraph{Phase 2: ODE flow sampling (Algorithm~\ref{alg2}).}
Given the weighted particles $\{(\bx_K^{(i)}, w_K^{(i)})\}_{i=1}^n$ from Phase~1, we
simulate the probability-flow ODE \eqref{eq:ode-ivp0} via Euler discretization.  The
velocity field at each step is estimated using the importance-weighted formula
\eqref{eq:velocity-field-IS} with the particles from Algorithm~\ref{alg1} serving as the
proposal samples.  To ensure numerical stability near $t = 1$ (where
expression~\eqref{eq:vf-expect-single-rmk} becomes singular), we stop at an
early-stopping time $T_\mathrm{end} = 1 - \epsilon$ for a small $\epsilon > 0$.  With $M$
discretization steps and step size $h = (T_\mathrm{end} - T_0)/M$, the update rule for
each new test particle $\bx_m^{(j)}$ is
\[
  \bx_{m+1}^{(j)} = \bx_{m}^{(j)} + h\,\hat{\bv}(t_m, \bx_m^{(j)}),
  \quad t_m = T_0 + mh, \quad 0 \le m < M,
\]
initialized from $\bx_0^{(j)} \sim N(0, \bI_d)$.  Under mild regularity conditions, the
distribution of $\bx_M^{(j)}$ approximates the target $\rho$ when the velocity estimation
error is sufficiently small.

\begin{algorithm}[H]
	\caption{ODE flow sampling using ALMC particles}
    \label{alg2}
	\begin{algorithmic}[1]
\State \textbf{Input:} weighted particles $\{(\bx_{K}^{(i)}, w_{K}^{(i)})\}_{i=1}^n$
       from Algorithm~\ref{alg1}; initial time $T_0=\epsilon$; terminal time
       $T_\mathrm{end} = 1 - \epsilon$; number of new test particles $n$; number of ODE
       steps $M$.
\State Compute step size $h = (T_\mathrm{end} - T_0)/M$.
\State Generate time grid: $t_m = T_0 + mh$ for $m = 0, 1, \ldots, M$.
\State Initialize test particles $\{\bx_0^{(j)}\}_{j=1}^N \sim N(0,\bI_d)$.
\For{$m = 0, 1, \ldots, M-1$}
  \State Estimate the velocity field at each test particle using \eqref{eq:velocity-field-IS}
         with ALMC samples $\{(\bx_K^{(i)}, w_K^{(i)})\}$ as the proposal:
         \[
           \hat{\bv}(t_m, \bx_m^{(j)})
             = \frac{\dot\alpha(t_m)}{\alpha(t_m)}\bx_m^{(j)}
               + c(t_m)
                 \frac{\sum_{i=1}^n \bx_K^{(i)}\,
                       g(t_m,\bx_m^{(j)},\bx_K^{(i)})\,w_K^{(i)}}
                      {\sum_{i=1}^n
                       g(t_m,\bx_m^{(j)},\bx_K^{(i)})\,w_K^{(i)}},
         \]
         where $c(t) = \dot\beta(t) - \frac{\dot\alpha(t)}{\alpha(t)}\beta(t)$ and
         $g(t,\bx,\by) = \exp\!\bigl(-\|\bx - \beta(t)\by\|^2 / (2\alpha(t)^2)\bigr)$.
  \State Update: $\bx_{m+1}^{(j)} = \bx_{m}^{(j)} + h\,\hat{\bv}(t_m, \bx_m^{(j)})$.
\EndFor
\State \textbf{Output:} approximate target samples $\{\bx_{M}^{(j)}\}_{j=1}^N$.
\end{algorithmic}
\end{algorithm}

\section{Error Analysis}\label{mse}

In this section we derive an error bound for our proposed method, focusing on the
mean-squared error (MSE) of the velocity-field estimator \eqref{eq:velocity-field-IS}.

\paragraph{Boundary behavior.}
Since $\alpha(1) = 0$ and $\beta(0) = 0$, the velocity-field
expression~\eqref{eq:vf-expect-single-rmk} has apparent singularities at $t = 0$ (via
the factor $1/\alpha(t)$) and at $t = 1$.  However, as shown in
\citep{ding2023samplingfollmerflow}, the velocity field extends continuously to the closed
interval $[0, 1]$ via the equivalent representation
\begin{equation*}
    \bv(t, \bx) = \frac{\dot{\beta}(t)}{\beta(t)}\,\bx
      + \frac{\alpha(t)\bigl[\alpha(t)\dot{\beta}(t)-\dot{\alpha}(t)\beta(t)\bigr]}{\beta(t)^2}
        \bE\bigl[\nabla\log\rho(\bx_1) \mid \bx_t = \bx\bigr],
    \quad t \in (0, 1].
\end{equation*}
In practice, to mitigate numerical instability near both endpoints, we restrict ODE
integration to the interval $[T_0, T_\mathrm{end}] \subset (0, 1)$ with $T_0$ and
$1 - T_\mathrm{end}$ both small and positive.

\paragraph{Sources of error.}
The global error of our method has two sources.  The first is the discrepancy between the
continuous ODE solution and its Euler discretization.  The second arises from approximating
the true velocity field by the importance-weighted Monte Carlo
estimator~\eqref{eq:velocity-field-IS}.  The discretization error has been analyzed in
\cite{ding2023samplingfollmerflow}; we therefore focus on the sensitivity of the ODE
system to perturbations in the velocity field, which quantifies the error introduced by
the Monte Carlo approximation.

\begin{assump}\label{assum1}
The target measure $\rho$ has a finite third moment and is absolutely continuous with
respect to the standard Gaussian measure.
\end{assump}

\begin{assump}\label{assum2}
There exists a random variable $\by \in \bR^d$ with law $p_{\by}$ satisfying
\begin{equation}\nn
\mathrm{supp}(p_{\by}) \subseteq B(0,R) := \{\bx \in \bR^d : \|\bx\| \le R\},
\quad R > 0,
\end{equation}
such that $\bx_1 = \by + \sigma\bepsilon$ with $\sigma^2 > 0$ and
$\bepsilon \sim N(0,\bI_d)$ independent of $\by$, i.e., the target density has the
Gaussian-mixture representation
\begin{equation}\label{assump}
\rho(\bx_1) = \int P_{\sigma^2}(\bx_1 - \by)\, p_{\by}(\by)\, \diff\by,
\end{equation}
where $P_{\sigma^2}$ denotes the Gaussian density with mean zero and variance $\sigma^2$.
\end{assump}

Under Assumption~\ref{assum1}, it was shown in
\cite{duan2026samplingstochasticinterpolantslangevinbased} that for every $t \in [0,1]$
the velocity field satisfies $\|\bv(t, \bx_t)\|_2 \le B(1 + \|\bx_t\|)$ and
$\|\nabla \bv(t, \bx_t)\|_{\mathrm{op}} \le G$, where $B>0$ and $G>0$ are constants
depending only on $\sigma$ and $R$.  Furthermore, Assumptions~\ref{assum1}
and~\ref{assum2} together imply that $\rho$ satisfies a log-Sobolev inequality, which
guarantees stability of the Euler-discretized ODE flow: samples produced by the Euler
approximation of \eqref{eq:ode-ivp0} using the true velocity field converge to $\rho$ as
the step size $h \to 0$.  Consequently, the pushforward distribution of the Euler scheme
approximates $\rho$ whenever the velocity estimation error is sufficiently small.

The following proposition shows that the Monte Carlo approximation of the velocity field
achieves the standard $\mathcal{O}(1/n)$ MSE rate as the particle count $n$ grows.

\begin{proposition}\label{prop3}
For each step $k \ge 0$, let $t_k$ denote the corresponding time grid point and set
$\alpha_k = \alpha(t_k)$, $\beta_k = \beta(t_k)$.  Define the Gaussian kernel
\begin{equation}\label{eq:gk-def}
    g_k(\bx, \bx_1)
    = \exp\!\left\{-\frac{\|\bx - \beta_k\bx_1\|^2}{2\alpha_k^2}\right\},
    \quad \bx,\bx_1 \in \bR^d.
\end{equation}
Let $\{\bx_1^{(i)}\}_{i=1}^n$ be the output particles of Algorithm~\ref{alg1} with
corresponding importance weights
\begin{equation}\label{eq:w-def}
    w(\bx_1) = \frac{\rho(\bx_1)}{\hat{p}_K(\bx_1)},
\end{equation}
where $\hat{p}_K$ is the estimated marginal density at step $K$ given by
\eqref{density2}.  Define the population velocity and its Monte Carlo approximation at a
fixed point $\bx \in \bR^d$ by
\begin{equation}\label{eq:v-star-hat}
\bv^\star_k(\bx)
  = \frac{\bE\bigl[\bx_1\, g_k(\bx,\bx_1)\, w(\bx_1)\bigr]}
         {\bE\bigl[g_k(\bx,\bx_1)\, w(\bx_1)\bigr]},
\qquad
\hat{\bv}_k(\bx)
  = \frac{\displaystyle\sum_{i=1}^n \bx_1^{(i)}\,
           g_k\!\left(\bx,\bx_1^{(i)}\right) w\!\left(\bx_1^{(i)}\right)}
         {\displaystyle\sum_{i=1}^n
           g_k\!\left(\bx,\bx_1^{(i)}\right) w\!\left(\bx_1^{(i)}\right)},
\end{equation}
where the expectation in $\bv^\star_k$ is taken over the marginal law of $\bx_k$ defined
by the iteration~\eqref{discrete}.  Under Assumptions~\ref{assum1} and~\ref{assum2}, for
any sample size $n \ge 1$ and step size satisfying $\delta < \min(\epsilon, \sigma^2)$,
we have
\begin{equation}\label{eq:mse-bound}
\bE\bigl\|\hat{\bv}_k(\bx) - \bv^\star_k(\bx)\bigr\|^2
  \le \mathcal{O}\!\left(\frac{1}{n}\right),
\end{equation}
for all $k \ge 0$.
\end{proposition}

\section{Numerical Experiments}\label{numerical}

In this section, we investigate the empirical performance of our algorithm for sampling from target distributions that may be unnormalized. To demonstrate its flexibility, we consider examples ranging from simple low-dimensional settings to challenging high-dimensional scenarios, showing that the method is capable of recovering multimodal distributions. Specifically, in Section~\ref{ldmix}, we apply our method to a two-dimensional Gaussian mixture with 20 components; in Section~\ref{hdmix}, we consider a high-dimensional Gaussian mixture with 5 components; and in Section~\ref{realmix}, we study high-dimensional multimodal distributions arising in large-scale real-world problems. Code can be obtained from the github
repository https://github.com/hanwenhuanghep/Flow-ODE-Sampling.

In addition to the Monte Carlo–based ODE approach of \cite{ding2023samplingfollmerflow}, which estimates the velocity field in~(\ref{velocity_int}) via Monte Carlo averaging with i.i.d.\ samples from the standard Gaussian distribution $N(0,\bI_d)$, we also compare against standard Hamiltonian Monte Carlo (HMC) \citep{duane1987, neal2011}. For the Gaussian mixture experiments in Sections~\ref{ldmix} and~\ref{hdmix}, ground-truth samples are drawn directly, and each method is evaluated using five metrics: (i)~$\ell_2$ mean error, (ii)~$\ell_2$ second-moment error, (iii)~Energy Distance \citep{szekely2013}, (iv)~maximum mean discrepancy (MMD) with an RBF kernel (with bandwidth selected via the median heuristic), and (v)~Sliced Wasserstein Distance (SWD) computed over 200 random projections. For the real-world application in Section~\ref{realmix}, we use both the unbiased U-statistic and the biased V-statistic of the kernelized Stein discrepancy as evaluation metrics \citep{doucet2022scorebaseddiffusionmeetsannealed}.

\subsection{Sampling from two-dimensional Gaussian mixture}\label{ldmix}

Our first example concerns sampling from a two-dimensional normal mixture model,
\begin{eqnarray}\label{mmodes}
f(\bx)&=&\sum_{i=1}^{I}\frac{w_i}{2\pi\sigma_i^2}\exp\left\{-\frac{1}{2\sigma_i^2}(\bx-\bmu_i)^T(\bx-\bmu_i)\right\}. 
\end{eqnarray}
We examine a mixture of 20 Gaussians with $I=20$, $\sigma_1=\cdots=\sigma_{20}=0.1$, and $w_1=\cdots=w_{20}=0.05$. The mean vectors are 
\begin{eqnarray}\nn
(\bmu_1,\cdots,\bmu_{20})&=&\left(\begin{array}{c}2.18~8.67~4.24~8.41~3.93~3.25~1.70~4.59~6.91~6.87\\
5.76~9.59~8.48~1.68~8.82~3.47~0.50~5.60~5.81~5.40\end{array}\right.\\\nn
&&\left.\begin{array}{c}5.41~2.70~4.98~1.14~8.33~4.93~1.83~2.26~5.54~1.69\\2.65~7.88~3.70~2.39~9.50~1.50~0.09~0.31~6.86~8.11\end{array}\right),
\end{eqnarray}
as studied in \cite{Liang01062001} and \cite{kou2006}. In this case, most local modes are more than five standard deviations from their nearest neighbor, ensuring that modes are well separated. Consequently, transitions between modes require trajectories passing through regions of near-zero probability—posing a significant challenge for sampling algorithms and making these mixtures stringent benchmarks. 

We applied our ALMC-ODE Algorithm to generate 10,000 samples from distribution (\ref{mmodes}). As shown in Figure~\ref{fig:kou_scatter}, the scatter plot of ALMC-ODE generated samples closely matches that of the true distribution by visiting all 20 mixture components with high frequency, producing empirical distributions that align well with the ground truth. In comparison, MC-ODE and HMC visit only a small subset of components within the same sampling budget: MC-ODE tends to explore only components located near the origin while struggling to reach those farther away, and HMC gets trapped in whichever single mode it encounters first. Table~\ref{tab:ldmix_metrics} further reports five distribution-distance metrics averaged over 20 independent runs: $\ell_2$ mean error, $\ell_2$ second-moment error, Energy Distance \citep{szekely2013}, MMD with RBF kernel, and Sliced Wasserstein Distance (SWD). Across all five metrics, ALMC-ODE substantially outperforms both MC-ODE and HMC. Notably, HMC achieves a high acceptance rate of $0.9659 \pm 0.0131$ yet exhibits large variance across seeds on all metrics—a signature of mode trapping, where each chain run gets stuck in a different single component.

\begin{figure}[ht]
\centering
\includegraphics[width=\textwidth]{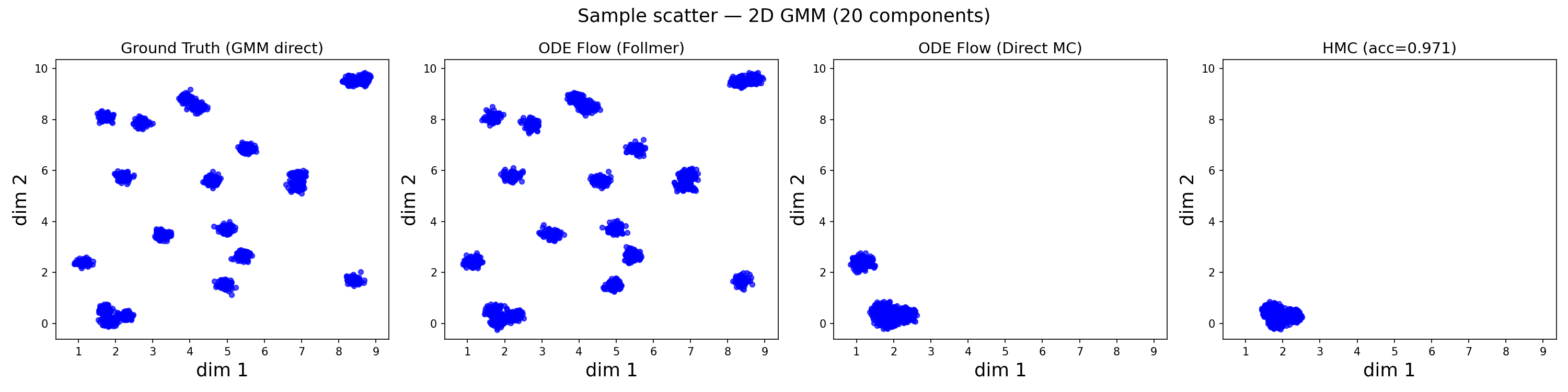}
\caption{Scatter plots of 10,000 samples from the 2-dimensional 20-component Gaussian mixture.}
\label{fig:kou_scatter}
\end{figure}

\begin{table}[ht]
\centering
\caption{Distribution-distance metrics (mean $\pm$ std) over 20 independent runs
         for the 2-dimensional 20-component Gaussian mixture. ALMC-ODE uses the F\"{o}llmer
         interpolant with AIS-weighted particles; MC-ODE uses per-step Monte Carlo velocity
         estimation; HMC is a single chain with step size $0.05$ and 10 leapfrog steps.
         Lower values indicate better agreement with the ground truth.}
\label{tab:ldmix_metrics}
\begin{tabular}{l ccc}
\toprule
\textbf{Metric} & \textbf{ALMC-ODE} & \textbf{MC-ODE} & \textbf{HMC} \\
\midrule
$\ell_2$ Mean Error        & $0.4403 \pm 0.0709$ & $5.0321 \pm 0.0793$ & $3.7861 \pm 1.6679$ \\
$\ell_2$ 2nd-Moment Error  & $6.6350 \pm 0.6393$ & $39.8293 \pm 0.6596$ & $33.8614 \pm 12.6232$ \\
Energy Distance            & $0.0864 \pm 0.0103$ & $5.4449 \pm 0.1596$ & $4.7977 \pm 1.7579$ \\
MMD (RBF)                  & $0.0105 \pm 0.0013$ & $0.6028 \pm 0.0125$ & $0.5356 \pm 0.1612$ \\
Sliced Wasserstein         & $0.4232 \pm 0.0298$ & $3.4296 \pm 0.0573$ & $3.1454 \pm 0.6198$ \\
\midrule
HMC Acceptance Rate        & ---                 & ---                 & $0.9659 \pm 0.0131$ \\
\bottomrule
\end{tabular}
\end{table}

\subsection{Sampling from high-dimensional Gaussian mixture}\label{hdmix}

Our second example considers the problem of drawing samples from a Gaussian mixture model (\ref{mmodes}) in $\mathbb{R}^d$ with $d = 100$, $I = 5$, equal mixture weights $w_i = 1/5$, and covariances $\Sigma_i = 0.1 \, \bI_d$ for all $i$. The component means are concentrated in the first two coordinates:
\begin{equation}\nn
    \bmu_1 = (10,10,0,\ldots,0)^\top, \quad
    \bmu_2 = (15,15,0,\ldots,0)^\top, \quad
    \bmu_3 = (5,15,0,\ldots,0)^\top,
\end{equation}
\begin{equation}\nn
    \bmu_4 = (15,5,0,\ldots,0)^\top, \quad
    \bmu_5 = (5,5,0,\ldots,0)^\top,
\end{equation}
with the remaining $d - 2$ coordinates set to zero. The minimum inter-mode Euclidean distance is $7.07 \approx 22\,\sigma$, making the mixture severely multimodal. This benchmark tests the ability of flow-based samplers to recover a well-separated multimodal distribution embedded in high dimensions. This experiment isolates the key advantage of the
probability-flow approach: the ability to transport mass across all modes simultaneously, bypassing the inter-mode barriers that trap gradient-based MCMC.

To obtain weighted samples, we run an annealed Langevin Markov chain with $T = 1000$ steps. The annealing schedule is specified by a step-size sequence $\{\delta_k\}_{k=1}^{T}$ linearly decreasing from $1.0$ to $0.1$, and an annealing sequence $\{\lambda(t_k)\}_{k=1}^{T}$ linearly increasing from $0$ to $1$. We generate $n = 10{,}000$ weighted samples for estimating the velocity field using an importance-weighted conditional expectation. We then propagate $N = 10{,}000$ fresh samples drawn from the standard Gaussian reference to the target $\rho$ by integrating an ODE. We adopt the F\"{o}llmer interpolating flow schedule for the drift.

HMC is run as a chain initialized from $\mathcal{N}(0, \bI_{100})$ with step size
$\varepsilon = 0.05$, $L = 10$ leapfrog steps, and a burn-in of 1,000 samples.  Both
methods produce 10,000 post-burn-in samples, and experiments are repeated 10 times.  Since MC-ODE performs extremely poorly in this high-dimensional setting, we report only the comparison with HMC.

Table~\ref{tab:summary} summarizes the mean $\pm$ standard deviation across replications.
ALMC-ODE consistently produces lower Energy Distance, MMD, and SWD, reducing Energy Distance by a factor of $\sim\!57\times$, MMD by
$\sim\!75\times$, and Sliced Wasserstein by $\sim\!7\times$ relative to HMC. HMC shows near-constant poor performance---a hallmark of mode trapping, with no
variance in where the chain gets stuck. ALMC-ODE is dramatically closer to the ground truth on all metrics.

\begin{table}[ht]
\centering
\caption{Summary statistics (mean $\pm$ std) over 10 replications for the 100-dimensional
         5-component Gaussian mixture. Lower values indicate better approximation of the target distribution.}
\label{tab:summary}
\begin{tabular}{l cc}
\toprule
\textbf{Metric} & \textbf{ALMC-ODE} & \textbf{HMC (single chain)} \\
\midrule
$\ell_2$ Mean Error        & $0.8402 \pm 0.1720$ & $7.075 \pm 0.035$ \\
$\ell_2$ 2nd-Moment Error  & $18.47 \pm 3.98$    & $134.4 \pm 0.77$  \\
Energy Distance            & $0.1036 \pm 0.0285$ & $5.926 \pm 0.052$ \\
MMD (RBF)                  & $0.00520 \pm 0.00172$ & $0.3895 \pm 0.0028$ \\
Sliced Wasserstein         & $0.0916 \pm 0.0122$ & $0.6612 \pm 0.0033$ \\
\midrule
HMC Acceptance Rate        & --- & $0.9739 \pm 0.0013$ \\
\bottomrule
\end{tabular}
\end{table}

Figure~\ref{fig:scatter} shows sample scatter plots in the first two dimensions for one realization: ground-truth
samples cover all five modes, ALMC-ODE samples spread across multiple modes, while HMC
samples cluster in a single mode.

\begin{figure}[ht]
\centering
\includegraphics[width=\textwidth]{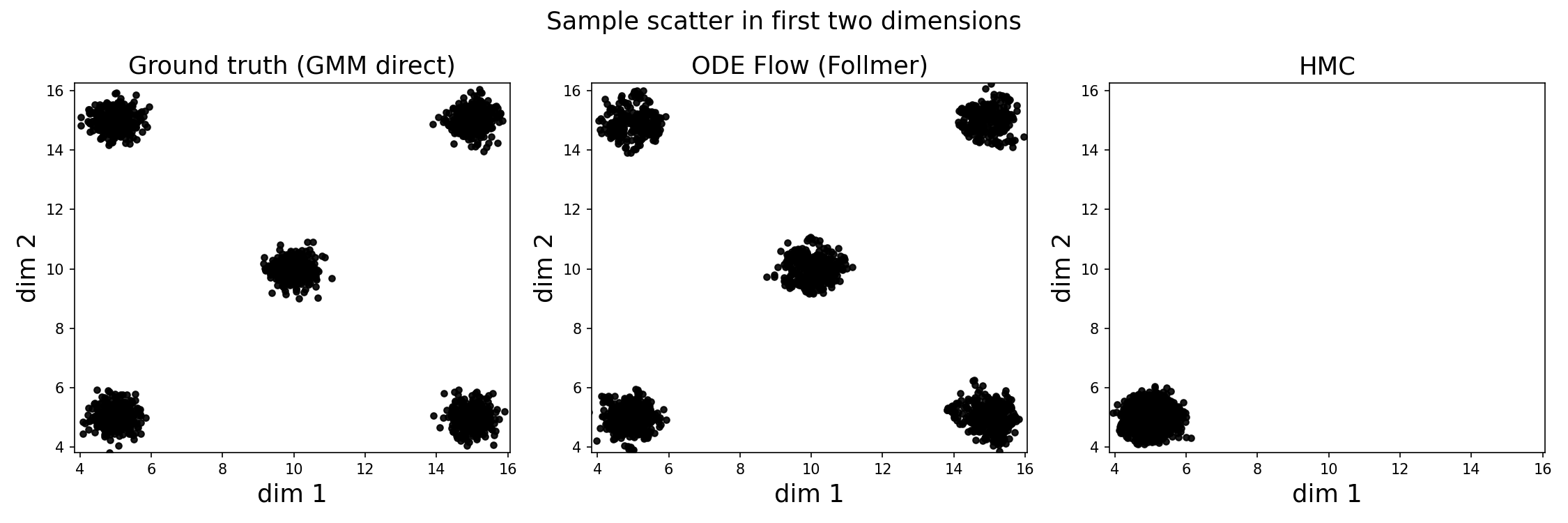}
\caption{Samples projected onto the first two dimensions for seed 81.
         Left: ground truth (all 5 modes); center: ALMC-ODE (multi-modal coverage);
         right: HMC (trapped in a single mode; acceptance rate 0.975).}
\label{fig:scatter}
\end{figure}

The F\"{o}llmer ODE integrates probability mass from $\mathcal{N}(0,\bI_d)$ proportionally
across \emph{all} modes, guided by ALMC-weighted particles.  Because the annealed
Langevin chain uses an interpolation schedule ($\lambda(t_k)$ from 0 to 1), it tunnels
between modes in the low-temperature phase and correctly assigns importance weights at
the final temperature.  The resulting ALMC-ODE samples therefore cover the full mixture
geometry.  The HMC acceptance rate is very high, yet the Energy Distance is roughly $57\times$ larger than that of ALMC-ODE.  This confirms that high
acceptance does not imply good mixing: the chain accepts nearly every proposal because it
explores one mode efficiently, but never crosses the $\approx\!22\,\sigma$ barrier to
reach other modes.  The near-zero variance of HMC metrics across replications reflects this
deterministic mode-trapping behavior.  The ALMC-ODE $\ell_2$ second-moment error ($18.47$) is substantially lower than HMC
($134.4$), because ALMC-ODE samples capture both inter-mode spread and intra-mode variance,
while HMC samples reflect only the variance of a single component.


\subsection{Allen--Cahn Field System}\label{realmix}

Our third example considers the stochastic Allen--Cahn model \citep{Berglund_2017}, a canonical and widely used framework for describing microscopic phase transitions in condensed matter systems. The model is defined in terms of a random field $\phi:[0,1] \rightarrow \mathbb{R}$ governed by the stochastic partial differential equation (SPDE)
\begin{eqnarray}\label{field}
\partial_t \phi = a \partial_s^2 \phi + a^{-1}(\phi - \phi^3) + \sqrt{2\beta^{-1}}\eta(t,s),
\end{eqnarray}
where $a>0$ is a parameter, $\beta$ denotes the inverse temperature, $s\in[0,1]$ is the spatial variable, and $\eta(t,s)$ is a spatio-temporal white noise. Dirichlet boundary conditions are imposed such that $\phi(0)=\phi(1)=0$. The invariant distribution of this SPDE is the Gibbs measure associated with the Hamiltonian
\begin{eqnarray}\label{Hamiltonian}
U(\phi) = \beta \int_0^1 \left\{ \frac{a}{2} (\partial_s \phi)^2 + \frac{1}{4a}(1 - \phi^2(s))^2 \right\} ds.
\end{eqnarray}
The first term in \eqref{Hamiltonian} penalizes spatial fluctuations of $\phi$, promoting alignment of the field in either the positive or negative direction. As a result, the Hamiltonian exhibits two global minima, denoted by $\phi^+$ and $\phi^-$, corresponding to configurations where $\phi \approx \pm 1$ (see Figure~\ref{fig:phifour_comparison}). 

To proceed computationally, we discretize the spatial domain $[0,1]$ using $d$ grid points $0=s_0<s_1<\cdots<s_d<s_{d+1}=1$. Let $\bx \in \mathbb{R}^d$ denote the vector of field values at the interior grid points, i.e., $x_i=\phi(s_i)$. This leads to a discretized target density of the form 
\begin{eqnarray}\label{dphifour}
\log \mu(\bx) = -\beta \left( \frac{a}{2\Delta s} \sum_{i=1}^{d+1} (x_i - x_{i-1})^2 + \frac{b \Delta s}{4} \sum_{i=1}^d (1 - x_i^2)^2 \right),
\end{eqnarray}
where $\Delta s=1/d$. In our experiments, we set $a=0.1,b=1/a=10$, and $\beta=20$. Due to the free energy barrier between $\phi^+$ and $\phi^-$,  traditional MCMC methods based on local updates fail to mix effectively, even over long timescales. In this scenario, the diffusive nature of our method is essential to prevent the learned flow from collapsing onto a single mode, enabling it instead to explore both global minima.

We discretize the field on a lattice of $d=64$ interior grid points with Dirichlet
boundary conditions $x_0=x_{d+1}=0$. We first use Algorithm~\ref{alg1} to produce $N=10{,}000$ importance-weighted particles in
        $T=10{,}000$ steps, with scale schedule $\delta_k$ linearly decreasing from
        $0.1$ to $0.001$ and annealing schedule $\lambda(t) = 1 - \ee^{-50t}$.  The probability-
        flow ODE (Algorithm~\ref{alg2}) then transports $n=1{,}000$ fresh Gaussian
        samples to the target using the F\"{o}llmer drift~\eqref{eq:velocity-field-IS} integrated with $N_t = 100$
        Euler steps. For MC-ODE, the velocity field is estimated by drawing $N_{\text{mc}}=10{,}000$ independent bridge samples at each
        step rather than reusing ALMC particles; importance weights are computed on the
        fly without annealing. For HMC,  we run a chain for $2{,}000$ burn-in steps with leapfrog steps $L=30$ and
        step size $\varepsilon=0.02$.  

Figure~\ref{fig:phifour_comparison} shows representative field configurations sampled by
each method.  Both $\phi^+$ and $\phi^-$ modes (configurations close to $\phi\approx+1$
and $\phi\approx-1$, respectively) are visually identifiable in the ALMC-ODE output; HMC
samples collapse to a single polarity, confirming mode trapping.  The MC-ODE estimator
produces numerically degenerate samples in this $d=64$ setting (its KSD is undefined).   

\begin{figure}[ht]
\centering
\begin{subfigure}[b]{0.32\textwidth}
  \centering
  \includegraphics[width=\textwidth]{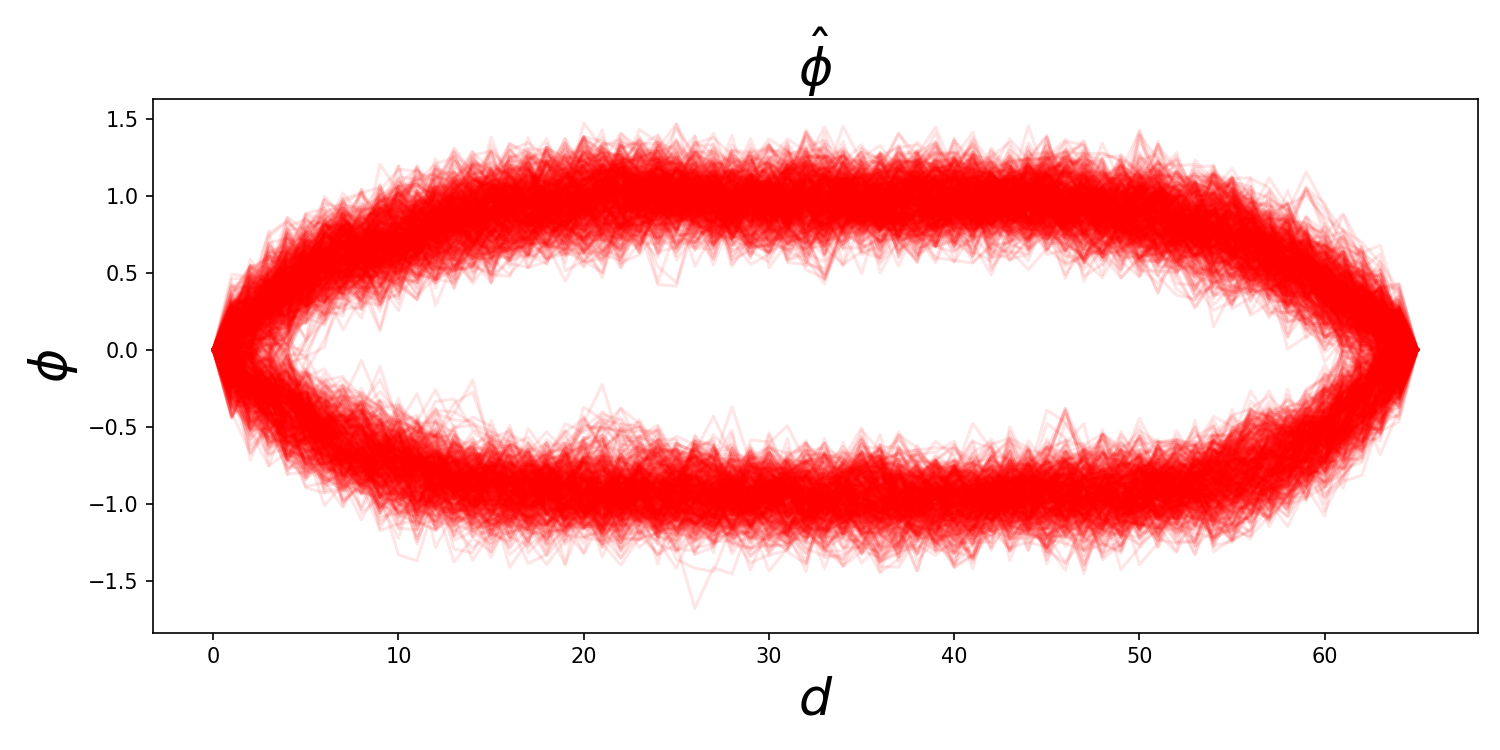}
  \label{fig:phifour_ode}
\end{subfigure}
\hfill
\begin{subfigure}[b]{0.32\textwidth}
  \centering
  \includegraphics[width=\textwidth]{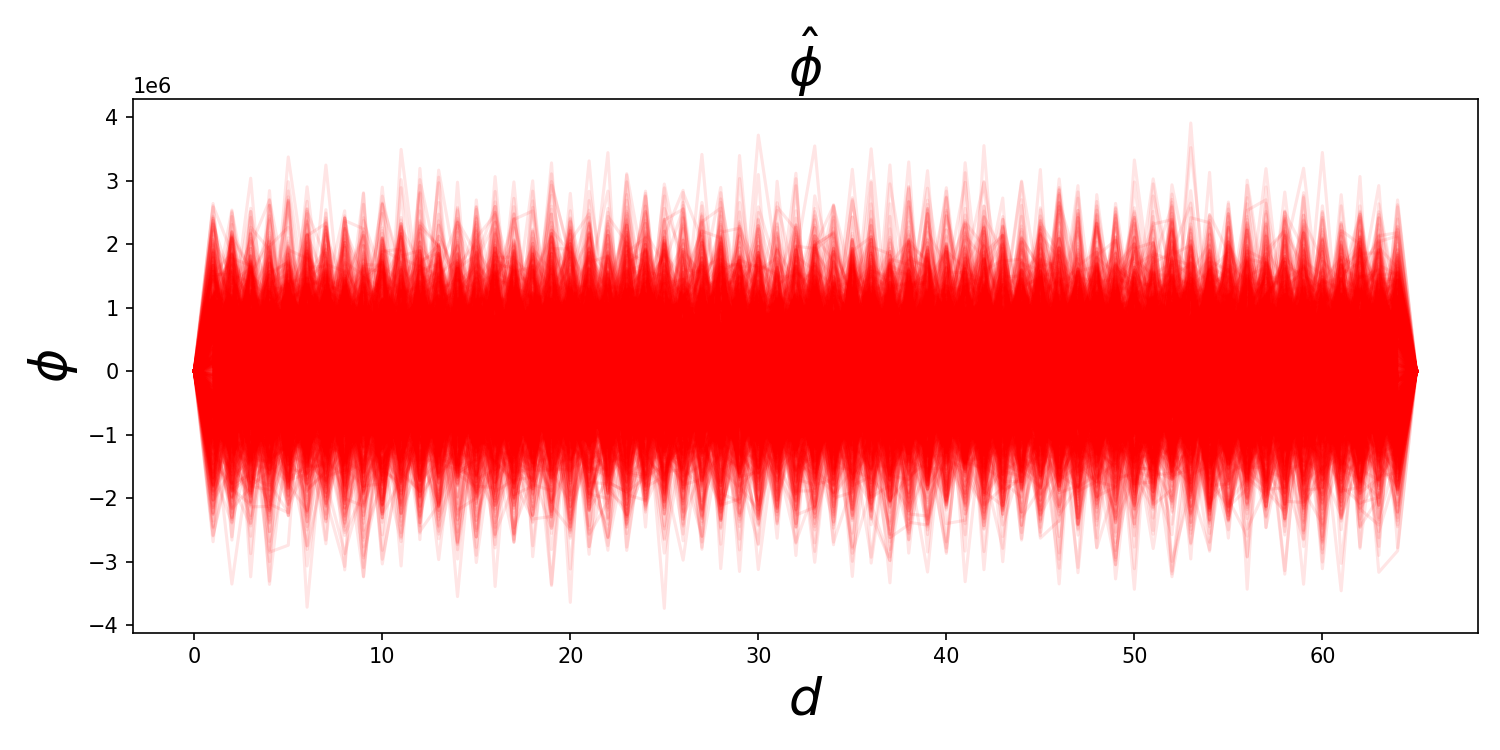}
  \label{fig:phifour_naive}
\end{subfigure}
\hfill
\begin{subfigure}[b]{0.32\textwidth}
  \centering
  \includegraphics[width=\textwidth]{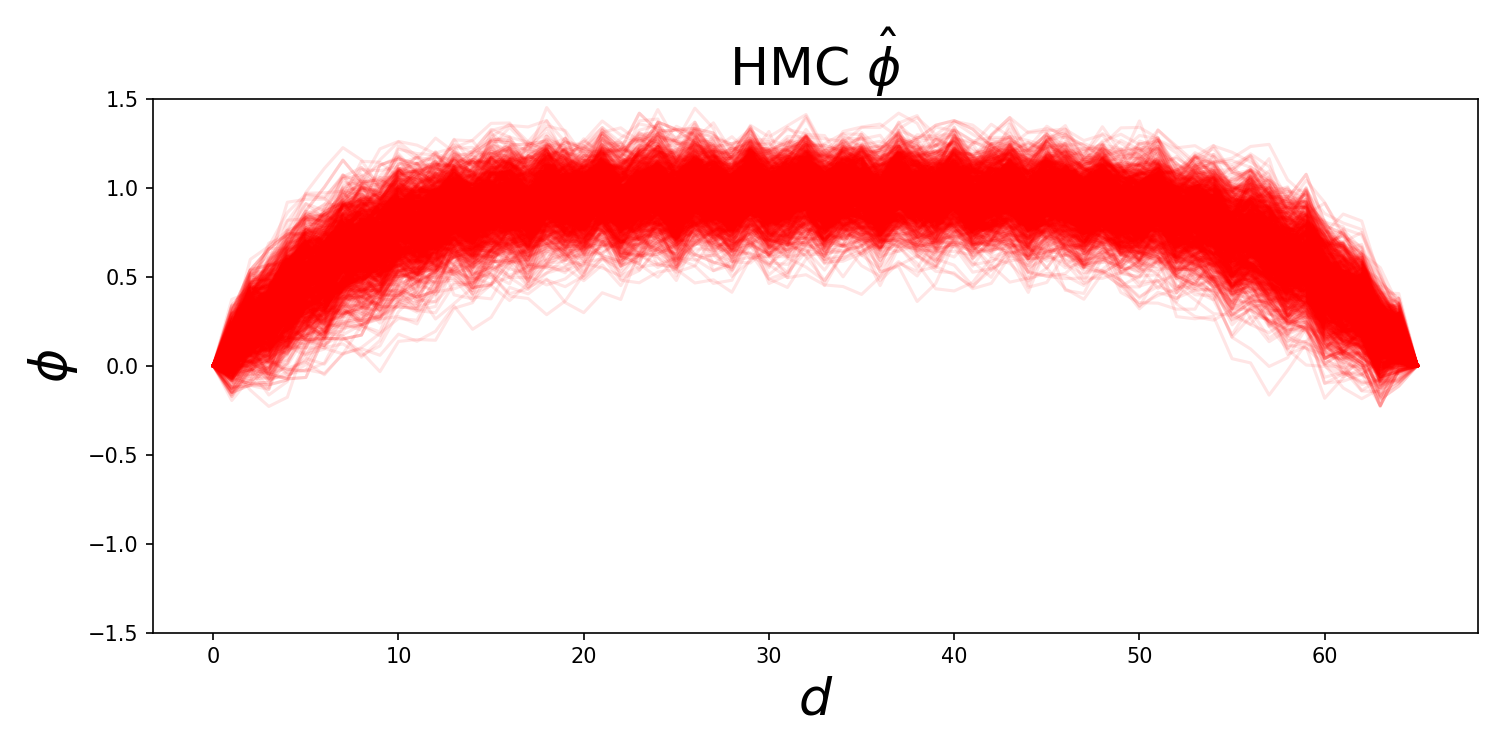}
  \label{fig:phifour_hmc}
\end{subfigure}
\caption{$d=64$ Allen--Cahn field configurations sampled by three methods.
         Each trace is a single field configuration $\bx\in\mathbb{R}^{64}$
         (padded with boundary zeros); $1{,}000$ samples are overlaid per panel.
         ALMC-ODE captures both $\phi^+$ and $\phi^-$ modes.
         MC-ODE produces degenerate near-zero configurations due to weight collapse.
         HMC is trapped in a single mode despite an acceptance rate of $0.936$.}
\label{fig:phifour_comparison}
\end{figure}

Table~\ref{tab:phifour_stein} reports the kernelized Stein discrepancy
\citep{liu2016kernelized} with the inverse multi-quadric kernel,
$k(x,x') = (1+\|x-x'\|^2)^{-1/2}$.  Both the unbiased U-statistic and the biased
V-statistic are listed; lower values indicate better agreement with the target $\mu$.

\begin{table}[ht]
\centering
\caption{Kernelized Stein discrepancy (KSD) for the $d=64$ Allen--Cahn target.
         Lower values indicate closer agreement with $\mu$.}
\label{tab:phifour_stein}
\begin{tabular}{lcc}
\toprule
\textbf{Method} & \textbf{KSD U-stat} & \textbf{KSD V-stat} \\
\midrule
ALMC-ODE     & $146.15$ & $217.23$ \\
MC-ODE       & NaN      & NaN      \\
HMC          & $759.00$ & $789.97$ \\
\bottomrule
\end{tabular}
\end{table}

ALMC-ODE successfully explores both the $\phi^+$ and $\phi^-$ phases of the field.
This is a direct consequence of the annealed Langevin chain interpolating between the
Gaussian reference and the target $\rho$ along a path that gradually reveals the
free-energy barrier; once the barrier is crossed during the diffusion phase, importance
reweighting corrects for the path.  In contrast, HMC with local proposals cannot tunnel
across the barrier on any practical timescale, even with a high acceptance rate of
$0.936$: the chain simply explores one mode efficiently while the other remains
inaccessible.

The MC-ODE velocity-field estimator draws fresh bridge samples at every integration
step and weights them by $\rho$ without annealing.  In $d=64$, the unnormalized
importance weights concentrate on an exponentially small fraction of particles---a
manifestation of the curse of dimensionality---so the weighted average collapses to a
single effective sample and the resulting KSD is undefined.  ALMC-ODE circumvents this
by distributing the bridging work across $T=10{,}000$ annealed steps,
keeping the effective sample size stable throughout.

\section{Conclusion}\label{conclusion}

We propose Annealed Langevin Monte Carlo for Flow ODE Sampling (ALMC-ODE), a method for sampling from unnormalized Boltzmann densities that integrates two key components: (i) an annealed Langevin Monte Carlo phase (Algorithm~\ref{alg1}), which generates importance-weighted particles that approximate the target distribution via a Jarzynski-based reweighting scheme, and (ii) a probability-flow ODE phase (Algorithm~\ref{alg2}), which transports fresh Gaussian samples to the target distribution by estimating the velocity field using these weighted particles.

The theoretical foundation of ALMC-ODE is built upon three main results. Proposition~\ref{prop1} establishes an exact reweighting identity for general time-inhomogeneous Markov chains, extending the classical framework of annealed importance sampling. Proposition~\ref{prop2} characterizes the optimal backward kernel that minimizes the variance of the Jarzynski weights, showing that it coincides with the time-reversal of the forward dynamics. Proposition~\ref{prop3} demonstrates that the Monte Carlo estimator of the velocity field achieves the standard $\mathcal{O}(1/n)$ mean squared error rate under mild regularity conditions.

Numerical experiments on several challenging benchmarks validate the practical effectiveness of ALMC-ODE. The results indicate that the proposed method provides a principled and competitive alternative to conventional MCMC approaches for highly multimodal target distributions, particularly in high-dimensional settings where energy barriers between modes render gradient-based samplers inefficient.

Several directions for future work warrant further investigation. On the theoretical side, extending Proposition~\ref{prop3} to a joint non-asymptotic error bound that simultaneously accounts for Langevin discretization error, ODE integration error, and finite-sample effects would provide clearer guidance on how to tune $K$, $\epsilon_k$, and $n$. On the algorithmic side, replacing the first-order Euler integrator with higher-order solvers, as well as incorporating adaptive annealing schedules (e.g., by equalizing the effective sample size between successive intermediate distributions \citep{neal2001annealed}), may help reduce both discretization bias and weight degeneracy. For very high-dimensional problems, substituting the kernel-based velocity estimator with a neural network trained on ALMC particles could improve scalability and further connect ALMC-ODE with diffusion-based generative modeling. Finally, extending the framework to applications such as Bayesian posterior inference and molecular free-energy estimation would broaden its practical impact.

\appendix

\section{Proof of Proposition~\ref{prop1}}\label{proof1}

\begin{proof}
By the recursive definition of $A_k$ in \eqref{iteration}, telescoping the increments
gives
\begin{eqnarray}\label{weight}
\exp(A_k)
  = \ee^{-V_k(\bx_k)+V_{0}(\bx_0)}
    \prod_{q=1}^k\frac{\nu_q(\bx_q,\bx_{q-1})}{\mu_q(\bx_{q-1},\bx_q)}.
\end{eqnarray}
The joint probability density of the path $(\bx_0,\ldots,\bx_k)$ under the forward
dynamics is
\begin{eqnarray}\nn
P(\bx_0,\ldots,\bx_k)
  = Z_{0}^{-1}\ee^{-V_{0}(\bx_0)}\prod_{q=1}^k\mu_{q}(\bx_{q-1},\bx_q).
\end{eqnarray}
For any measurable $f$, multiplying through by $\exp(A_k)$ and integrating over the path
gives
\begin{eqnarray}\nn
\bE[f(\bx_k)e^{A_k}]
  &=& \int f(\bx_k)\,Z_{0}^{-1}\ee^{-V_k(\bx_k)}
      \prod_{q=1}^k\nu_{q}(\bx_{q},\bx_{q-1})\,\diff\bx_0\cdots \diff\bx_k \\\nn
  &=& Z_{0}^{-1}\int f(\bx_k)\,\ee^{-V_k(\bx_k)}\,\diff\bx_k
   = Z_{0}^{-1}Z_{k}\,\bE_{k}[f(\bx_k)],
\end{eqnarray}
where the second equality uses the fact that $\int\nu_q(\bx_q,\bx_{q-1})\,\diff\bx_{q-1}
= 1$ for each $q$, so the integrals over $\bx_0,\ldots,\bx_{k-1}$ collapse to one.
Taking $f \equiv 1$ yields $\bE[e^{A_k}] = Z_{0}^{-1}Z_{k}$.  Dividing the general
identity by this normalization constant gives the first claim in \eqref{adjust}, and the
second claim $Z_k = Z_0\,\bE[e^{A_k}]$ follows immediately.
\end{proof}

\section{Proof of Proposition~\ref{prop2}}\label{proof2}

\begin{proof}
We show that the backward kernel \eqref{back} is a valid probability kernel and that it
minimizes the variance of $\exp(A_k)$.

\paragraph{Validity.}
Substituting \eqref{back} into the normalization condition verifies that $\nu_k^{\mathrm{opt}}$
integrates to one:
\begin{eqnarray}\nn
\int\nu_k^{\mathrm{opt}}(\bx_k,\bx_{k-1})\,\diff\bx_{k-1}
  = \int\frac{p_{k-1}(\bx_{k-1})\,\mu_k(\bx_{k-1},\bx_k)}{p_k(\bx_k)}\,\diff\bx_{k-1}
  = \frac{p_k(\bx_k)}{p_k(\bx_k)} = 1,
\end{eqnarray}
where the second equality uses the Chapman–Kolmogorov relation
$p_k(\bx_k) = \int p_{k-1}(\bx_{k-1})\mu_k(\bx_{k-1},\bx_k)\,\diff\bx_{k-1}$.

\paragraph{Variance minimization.}
From Proposition~\ref{prop1}, the reweighting identity \eqref{adjust} holds for any
valid pair $(\mu_k, \nu_k)$.  The variance of the estimator is governed by the variance
of $\exp(A_k)$.  By the chain rule and the law of total variance applied to the joint
path distribution,
\begin{eqnarray}\nn
\var_Q[\exp(A_k)]
  = \var_{q_k}[w(\bx_k)]
    + \bE_{q_k}\bigl[\var_{Q(\cdot|\bx_k)}[\exp(A_k)]\bigr],
\end{eqnarray}
where $q_k$ denotes the marginal of $\bx_k$ under the backward measure $Q$ and
$w(\bx_k) = \tilde{\pi}_k(\bx_k)/p_k(\bx_k)$.  Substituting \eqref{back} into
\eqref{weight} yields
\begin{eqnarray}\nn
\exp(A_k)
  &=& \frac{Z_k\,\tilde{\pi}_{k}(\bx_k)}{Z_0\,\tilde{\pi}_{0}(\bx_0)}
      \prod_{q=1}^{k}\frac{\nu_{q}^{\mathrm{opt}}(\bx_{q},\bx_{q-1})}
                          {\mu_{q}(\bx_{q-1},\bx_{q})} \\\nn
  &=& \frac{Z_k\,\tilde{\pi}_{k}(\bx_k)}{Z_0\,\tilde{\pi}_{0}(\bx_0)}
      \prod_{q=1}^{k}\frac{p_{q-1}(\bx_{q-1})}{p_{q}(\bx_{q})} \\\nn
  &=& \frac{Z_k\,\tilde{\pi}_{k}(\bx_k)}{Z_0\,p_{k}(\bx_k)},
\end{eqnarray}
where the telescoping product collapses using $p_0(\bx_0) = \tilde{\pi}_0(\bx_0)$.
The conditional variance term $\var_{Q(\cdot|\bx_k)}[\exp(A_k)]$ is then zero—given
$\bx_k$, $\exp(A_k)$ is a deterministic function—so the total variance reduces to
$\var_{q_k}[w(\bx_k)]$, which depends only on the marginal weights and not on the choice
of backward kernel.  Hence \eqref{back} achieves the minimum possible variance, equal to
$\var_{q_k}[\tilde{\pi}_k(\bx_k)/p_k(\bx_k)]$.
\end{proof}

\section{Proof of Proposition~\ref{prop3}}\label{proof3}

\begin{proof}
We derive an $\mathcal{O}(1/n)$ MSE bound for the velocity-field approximation produced
by Algorithm~\ref{alg1}.  Throughout, $\bx \in \bR^d$ is a fixed evaluation point and
all expectations are over the joint law of the particles $\{\bx_k^{(i)}\}$ defined by
the iteration~\eqref{discrete}.

\paragraph{Step 1: Shorthand notation.}
Recall from \eqref{eq:gk-def} and \eqref{eq:w-def} that
\[
g_k(\bx, \bx_1) = \exp\!\bigl\{-\|\bx - \beta_k\bx_1\|^2 / (2\alpha_k^2)\bigr\}
\quad\text{and}\quad
w(\bx_1) = \rho(\bx_1)/\hat{p}_K(\bx_1).
\]
Introduce the shorthands
\begin{equation}\label{eq:bc-def}
    \bb_k(\bx_1) = \bx_1\, g_k(\bx, \bx_1), \qquad
    c_k(\bx_1)  = g_k(\bx, \bx_1),
\end{equation}
and the weighted sample averages
\begin{equation}\label{eq:overbars}
    \overline{\bb_k w}
      = \frac{1}{n}\sum_{i=1}^n \bb_k\!\left(\bx_1^{(i)}\right) w\!\left(\bx_1^{(i)}\right),
    \qquad
    \overline{c_k w}
      = \frac{1}{n}\sum_{i=1}^n c_k\!\left(\bx_1^{(i)}\right) w\!\left(\bx_1^{(i)}\right).
\end{equation}
With this notation, the population and empirical velocity estimators from
\eqref{eq:v-star-hat} become
\begin{equation}\nn
    \bv^\star_k(\bx)
      = \frac{\bE[\bb_k(\bx_1)\,w(\bx_1)]}{\bE[c_k(\bx_1)\,w(\bx_1)]},
    \qquad
    \hat{\bv}_k(\bx)
      = \frac{\overline{\bb_k w}}{\overline{c_k w}}.
\end{equation}

\paragraph{Step 2: Decomposing the error.}
Write $\mu_{\bb} = \bE[\bb_k w]$ and $\mu_c = \bE[c_k w]$ for the population means.
Using the identity $a/b - c/d = (ad - bc)/(bd)$,
\begin{equation}\nn
    \hat{\bv}_k(\bx) - \bv^\star_k(\bx)
      = \frac{\mu_c\,\overline{\bb_k w} - \mu_{\bb}\,\overline{c_k w}}
             {\mu_c\,\overline{c_k w}}.
\end{equation}
Adding and subtracting $\mu_{\bb}\mu_c$ in the numerator and applying the triangle
inequality $(a+b)^2 \le 2(a^2+b^2)$ followed by Cauchy–Schwarz gives
\begin{eqnarray}\nn
\bE\bigl\|\hat{\bv}_k(\bx) - \bv^\star_k(\bx)\bigr\|^2
  &\le& \frac{\|\mu_{\bb}\|^2}{\mu_c^2}
        \,\bE\!\left[\frac{(\overline{c_k w} - \mu_c)^2}{\overline{c_k w}^2}\right]
      + \bE\!\left[\frac{\|\overline{\bb_k w} - \mu_{\bb}\|^2}{\overline{c_k w}^2}\right].
\end{eqnarray}
Both terms are handled identically; we bound the second in detail.

\paragraph{Step 3: Applying Jensen's inequality to the denominator.}
Since $x \mapsto 1/x^2$ is convex for $x > 0$, Jensen's inequality applied to the
empirical average $\overline{c_k w}$ yields
\begin{equation}\nn
    \frac{1}{\overline{c_k w}^2}
    \le \frac{1}{n}\sum_{i=1}^n \frac{1}{\bigl[c_k(\bx_1^{(i)})\,w(\bx_1^{(i)})\bigr]^2}.
\end{equation}

\paragraph{Step 4: Bounding the numerator variance.}
Substituting the Jensen bound and expanding the squared norm, the second term becomes
\begin{eqnarray}\nn
\bE\!\left[\frac{\|\overline{\bb_k w} - \mu_{\bb}\|^2}{\overline{c_k w}^2}\right]
  &\le&
  \bE\!\left\{
    \left\|\frac{1}{n}\sum_{i=1}^n
      \bigl[\bb_k(\bx_1^{(i)})w(\bx_1^{(i)})-\mu_{\bb}\bigr]\right\|^2
    \cdot
    \frac{1}{n}\sum_{j=1}^n
      \frac{1}{\bigl[c_k(\bx_1^{(j)})\,w(\bx_1^{(j)})\bigr]^2}
  \right\}.
\end{eqnarray}
Expanding the squared average and using the conditional independence of the particles
$\{\bx_1^{(i)}\}$, the cross-terms ($i \ne j$) vanish after taking expectations,
leaving diagonal and off-diagonal contributions:
\begin{eqnarray}\nn
&=&\frac{1}{n^2}\,
     \bE\!\left\{
       \frac{\|\bb_k(\bx_1^{(1)})w(\bx_1^{(1)})-\mu_{\bb}\|^2}
            {\bigl[c_k(\bx_1^{(1)})\,w(\bx_1^{(1)})\bigr]^2}
     \right\}\\\nn
&&+\;\frac{n-1}{n^2}\,
     \bE\!\left\|\bb_k(\bx_1^{(2)})w(\bx_1^{(2)})-\mu_{\bb}\right\|^2
     \cdot
     \bE\!\left\{
       \frac{1}{\bigl[c_k(\bx_1^{(1)})\,w(\bx_1^{(1)})\bigr]^2}
     \right\}.
\end{eqnarray}
Both terms involve the common factor
$\bE\!\bigl\{1/\bigl[c_k(\bx_1^{(1)})w(\bx_1^{(1)})\bigr]^2\bigr\}$, which we bound
next.

\paragraph{Step 5: Bounding the reciprocal-weight moment.}
Recalling $c_k(\bx_1) = g_k(\bx,\bx_1)$ and $w(\bx_1) = \rho(\bx_1)/\hat{p}_K(\bx_1)$,
\begin{eqnarray}\nn
\bE\!\left\{\frac{1}{\bigl[c_k(\bx_1^{(1)})\,w(\bx_1^{(1)})\bigr]^2}\right\}
  &=& \int
      \frac{\exp\!\left\{\dfrac{\|\bx - \beta_k\bx_1\|^2}{\alpha_k^2}\right\}
            \hat{p}_K(\bx_1)^2}
           {\rho(\bx_1)^2}
      \,\diff\bx_1 \\\nn
  &\le& \int \bE_{\by}\!\left[
          \exp\!\left\{
            \frac{\|\bx - \beta_k\bx_1\|^2}{\alpha_k^2}
            + \frac{\|\bx_1 - \by\|^2}{\sigma^2}
          \right\}
        \right]
        \hat{p}_K(\bx_1)^2
        \,\diff\bx_1,
\end{eqnarray}
where the inequality uses Assumption~\ref{assum2} to bound $1/\rho(\bx_1)^2$ via the
Gaussian-mixture structure of $\rho$ and Jensen's inequality.  This integral is finite
provided $\delta < \min(\epsilon, \sigma^2)$, since the quadratic exponent is dominated
by the Gaussian tails under that condition.  Denote this finite constant by $C_0 > 0$.

\paragraph{Step 6: Combining the bounds.}
Under Assumption~\ref{assum2}, $\bx_1$ has bounded support and $g_k \le 1$, so
$\|\bb_k(\bx_1)\| \le R$ for all $\bx_1$.  Moreover,
$\bE\bigl[\|\bb_k(\bx_1)w(\bx_1) - \mu_{\bb}\|^2\bigr] \le \mathcal{O}(1)$ follows from the bounded
support and the fact that the effective sample size ensures $\bE[w(\bx_1)^2] \le C_1$
for some finite constant $C_1$.  Combining the bounds from Steps 4 and 5,
\begin{equation}\nn
\bE\!\left[\frac{\|\overline{\bb_k w} - \mu_{\bb}\|^2}{\overline{c_k w}^2}\right]
  \le \frac{C_0\,\mathcal{O}(1)}{n^2}
    + \frac{(n-1)\,C_0\,\mathcal{O}(1)}{n^2}
  = \mathcal{O}\!\left(\frac{1}{n}\right).
\end{equation}
An identical argument shows that the first term in Step 2 is also $\mathcal{O}(1/n)$,
which completes the proof of \eqref{eq:mse-bound}.
\end{proof}

\bibliographystyle{chicago} 
\bibliography{biblist}

\end{document}